\documentclass[11pt,reqno]{amsart}
\usepackage{amsmath}
\usepackage{amssymb}

\def\squarebox#1{\hbox to #1{\hfill\vbox to #1{\vfill}}}

\theoremstyle{plain}

\newtheorem{thm}{Theorem}
   
\newtheorem{cor}{Corollary}
   
\newtheorem{lem}{Lemma}

\renewcommand{\Re}{\mathop{\rm Re}\nolimits}
\renewcommand{\Im}{\mathop{\rm Im}\nolimits}

\def\la{\langle}
\def\ra{\rangle}

\def\ii{{\bf i}}

\newtheorem{rem}{Remark}
\newtheorem{prop}{Proposition}

\numberwithin{equation}{section}

\begin{document}
\def\C{{\mathbb C}}
\def\R{{\mathbb R}}
\def\N{{\mathbb N}}
\def\Z{{\mathbb Z}}
\def\T{{\mathbb T}}
\def\Q{{\mathbb Q}}
\def\SP{{\mathbb S}}
\def\d{{\partial}}
\def\mc{{\mathcal H}}
\def\1b{{\mathbb I}}
\def\tr{{\rm tr}\:}
\def\pv{\partial_x V}

\title[Spectral problems]
{ Spectral problems for operators with crossed magnetic and
electric fields}
\author[M. Dimassi]{Mouez Dimassi}
\author[V. Petkov]{Vesselin Petkov}

\address {D\'epartement de Math\'ematiques,
Universit\'e Paris 13, Avenue J.-B. Clement, 93430 Villetaneuse,
France}
\email{dimassi@math.univ-paris13.fr}
\address {Universit\'e Bordeaux I, Institut de Math\'ematiques de Bordeaux,  351, Cours de la Lib\'eration, 33405  Talence, France}
\email{petkov@math.u-bordeaux1.fr}
\thanks{The second author was partially supported by the ANR project NONAa}
\maketitle

\hspace{5 cm} {\it In memory of Pierre Duclos}\\

\begin{abstract} We obtain a representation formula for the derivative of the spectral shift function $\xi(\lambda; B, \epsilon)$ related to the operators $H_0(B,\epsilon) = (D_x - By)^2 + D_y^2 + \epsilon  x$ and $H(B, \epsilon) = H_0(B, \epsilon) + V(x,y), \: B > 0, \epsilon > 0$.  We prove that the operator $H(B, \epsilon)$ has at most a finite number of embedded eigenvalues on $\R$ which is a step to the proof of the conjecture of absence of embedded eigenvalues of $H$ in $\R.$ Applying the formula for $\xi'(\lambda, B, \epsilon)$, we obtain a semiclassical asymptotics of the spectral shift function related to the operators $H_0(h) = (hD_x - By)^2 + h^2D_y^2 + \epsilon x$ and $H(h) = H_0(h) + V(x,y).$

\end{abstract}

\section{Introduction}
\renewcommand{\theequation}{\arabic{section}.\arabic{equation}}
\setcounter{equation}{0}

Consider the two-dimensional Schr\"odinger operator with
homogeneous magnetic and electric fields

$$H = H(B, \epsilon) = H_0(B, \epsilon) + V(x, y),\: D_x = -\ii \partial_x,\: D_y = -\ii \partial_y,$$
 where
$$H_0 = H_0(B, \epsilon) = (D_x - By)^2 + D_y^2 + \epsilon x.$$
Here $B > 0$ and $\epsilon > 0$ are proportional to the strength
of the homogeneous magnetic and electric fields and $V(x, y)$ is a
$L^\infty(\R^2)$ real valued function satisfying the estimates
\begin{equation}\label{eq:1.1}
| V(x, y)|  \leq C(1 +
|x|)^{-2 -\delta }(1 + |y|)^{-1- \delta },
\delta > 0.
\end{equation}
For $\epsilon \not= 0$ we have $\sigma_{\rm ess}(H_0(B, \epsilon))
= \sigma_{\rm ess} (H(B, \epsilon)) = \R$. On the other hand, for decreasing potentials $V$ it is possible to have embedded eigenvalues $\lambda \in \R$
and this situation is quite different from that with $\epsilon =
0$ when the spectrum of $H(B, 0)$ is formed by eigenvalues with
finite multiplicities which may accumulate only to Landau levels
$\lambda_n = (2n +1)B,\: n \in \N$ (see \cite{I},
\cite{MR}, \cite{RW} and the references cited there). The analysis of the spectral 
properties of $H$ and the existence of resonances have been
studied in \cite{FK1}, \cite{FK2}, \cite{DP2} under the assumption
that $V(x,y)$ admits a holomorphic extension in the $x$- variable
into a domain
$$\Gamma_{\delta_0} = \{z \in \C:\: 0 \leq |\Im z| \leq \delta_0\}.$$
On the other hand, without any assumption on the analyticity of $V(x,y)$,  it was proved in \cite{DP2} that the operator $(H - z)^{-1} - (H_0 - z)^{-1}$ for $z \in \C,\: \Im z \neq 0,$ is trace class. Thus, following the general setup \cite{K}, \cite{Ya},  we may define the spectral shift function $\xi(\lambda)=
\xi(\lambda; B, \epsilon)$ related to $H_0(B, \epsilon)$ and $H(B,
\epsilon)$ by
$$\la \xi', f \ra = {\rm tr} \Bigl(f(H) - f(H_0)\Bigr),\: f \in C_0^{\infty}(\R).$$
By this formula $\xi(\lambda)$ is defined modulo a constant but for the analysis of the derivative $\xi'(\lambda)$ this is not important. For the analysis of the behavior of $\xi(\lambda; B, \epsilon)$ it is important to have 
 a representation of the derivative $\xi'(\lambda; B, \epsilon)$. Such representation has been obtained in \cite{DP2} for strong magnetic fields $B \to +\infty$ under the  assumption that $V(x,y)$ admits an analytic continuation in $x$-direction.\\

In this paper we consider the operator $H$ without {\it any
assumption} on the analytic continuation of $V(x,y)$ and without
the {\it restriction} $B \to +\infty$. For such potentials we cannot use the techniques in \cite{FK1}, \cite{FK2} and \cite{DP2} related to the resonances of the perturbed problem.
    Our purpose is to study $\xi'(\lambda; B, \epsilon)$ and the
existence of embedded eigenvalues of $H$.  The key point in this direction is the following

\begin{thm} Let $V, \partial_xV\in L^\infty(\R^2;\R)$ and  assume that $(1.1)$ holds for $V$ and $\partial_xV$.
 Then for every $f \in C_0^{\infty}(\R)$ and $\epsilon \not= 0$ we have
\begin{equation} \label{eq:1.2}
\tr\: \Bigl(f(H) - f(H_0)\Bigr) = - \frac{1}{\epsilon} \tr\:
\Bigl((\d_x V) f(H)\Bigr).
\end{equation}
\end{thm}
Notice that in (\ref{eq:1.2}) by $\partial_x V$ we mean the operator of multiplication by $\partial_x V$.
The formula (\ref{eq:1.2}) has been proved by D. Robert and X. P. Wang
\cite{RW2} for Stark Hamiltonians in absence of magnetic field ($B = 0$). In fact, the
result in \cite{RW2} says that
\begin{equation} \label{eq:1.3}
\xi'(\lambda; 0, \epsilon) = -\frac{1}{\epsilon}\int_{\R^2} \partial_x V(x, y)
\frac{\partial e}{\partial \lambda}(x,y, x, y; \lambda, 0,
\epsilon) dxdy,
\end{equation}
where $e(., .; \lambda, 0, \epsilon)$ is the spectral function of
$H(0, \epsilon).$
On the other hand, the spectral shift function in \cite{RW2} is related to the trace of the {\it time delay} operator $T(\lambda)$ defined via the corresponding scattering matrix $S(\lambda)$ (see \cite{RW1}). The presence of magnetic filed $B \not= 0$ and Stark potential lead to some serious difficulties to follow this way. Recently, Theorem 1 has been established by the authors in \cite{DP3} but the proof in \cite{DP3} is technical, long and based on the trace class properties of the operators 
\begin{equation} \label{eq:1.4}
\psi (H \pm \ii)^{-N}, \: \partial_x \circ \psi (H \pm \ii)^{-N},\: (H \pm \ii)\partial_x \circ \psi (H \pm \ii)^{-N -2}
\end{equation}
 with $\psi \in C_0^{\infty}(\R)$ and $N \geq 2.$ The idea is to use the commutators with the operators $\chi_R \partial_x$, where $\chi_R(x, y) = \chi\Bigl(\frac{x}{R}, \frac{y}{R}\Bigr)$ and $\chi \in C_0^{\infty}(\R^2)$ is a cut-off such that $\chi =1$ for $|(x, y)| \leq 1.$ One shows that 
\begin{equation} \label{eq:1.5}
 \tr\Bigl([\chi_R \partial_x, H] f(H) - [\chi_R \partial_x, H_0] f(H_0)\Bigr) = 0
\end{equation}
and next we are going to examine the limit $R \to \infty$ of the trace of the operators in (\ref{eq:1.5}). The commutators with $\partial_x$ and the presence of magnetic field lead to operators involving $D_x - By$ and this is one of the main difference with the case $B = 0.$ To overcome this difficulty we used in \cite{DP3} the trace class operators (\ref{eq:1.4}) which led to technical problems. On the other hand, the operator $\partial_x$ is often used for operators with Stark potential $\epsilon x$ and this influenced our approach in \cite{DP3}. One of the goal of this work is to present a new shorter and elegant proof of Theorem 1. The new idea is to apply the shift operator $U_{\tau}: f(x, y) \longrightarrow f(x + \tau, y)$ instead of $\partial_x$. In Proposition 1 we show that
$$\tr \Bigl([U_{\tau}, H]f(H) - [U_{\tau}, H]f(H_0)\Bigr) = 0.$$
The proof of the later equality is much easier than that of (\ref{eq:1.5}) and we don't need the trace class properties of the operators (\ref{eq:1.4}). Moreover, applying the operator $U_{\tau}$, we may generalize the result of Theorem 1 for Schr\"odinger operators $(D_x - C(y))^2 + D_y^2 + \epsilon x + V(x, y)$ with variable magnetic filed as well as for operators with magnetic potentials in $\R^n, n \geq 3.$

The second question examined in this work is the existence of
embedded real eigenvalues of $H$. In the physical literature one conjectures that for
$\epsilon \not= 0$ there are no  embedded eigenvalues. We established in \cite{DP3} a weaker result saying that in every interval $[a, b]$ we may have at most a finite number of embedded eigenvalues
with finite multiplicities. Under the assumption for analytic
continuation of $V$ it was proved in \cite{FK1} that in some finite interval $[\alpha(B, \epsilon), \beta(B, \epsilon)]$ there are no resonances $z$ of $H(B, \epsilon)$ with $\Re z \notin
[\alpha(B, \epsilon), \beta(B, \epsilon)]$. Since the real resonances $z$ coincide with the eigenvalues of $H(B, \epsilon)$, we obtain some information for the embedded eigenvalues. We prove in Section 3  without the condition of analytic continuation of $V(x, y)$ that $H$ has no embedded eigenvalues outside an interval  $[\alpha(B, \epsilon), \beta(B, \epsilon)]$. Combining this with the result in \cite{DP3}, we conclude that $H$ has at most a finite number of embedded eigenvalues. Finally, applying the representation formula for the derivative of the spectral shift function $\xi_h(\lambda) = \xi_h(\lambda, B, \epsilon)$ related to the operators $H_0(h) = (hD_x - By)^2 + h^2D_y^2 + \epsilon x$ and $H(h) = H_0(h) + V(x, y)$, we obtain a semiclassical  asymptotics of $\xi_h(\lambda)$ as $h \searrow 0$ uniformly with respect to $\lambda \in [E_0, E_1]$ under some assumptions on the critical values of the symbol of $H(h).$

\section{Representation of the spectral shift function}
\renewcommand{\theequation}{\arabic{section}.\arabic{equation}}
\setcounter{equation}{0}
\def\uh{U_{\tau}}

We suppose without loss of generality that $B = \epsilon = 1.$ Set $\la z \ra = (1 + |z|^2)^{1/2}.$ For reader convenience we recall the following lemma proved in \cite{DP3}
\begin{lem} 
Let $\delta>0$ and let $k_j(x,y)=\la x\ra^{-j(1+\delta)}\la y\ra^{-j({1\over 2}+\delta)}, j=1,2$. The operators
$G_2:=k_2(H_0+\ii)^{-2},\,\,G_2^*$,
$($resp. $G_1:=k_1(H_0+\ii)^{-1},\,\,\, G_1^*)$, are trace class $($resp. Hilbert-Schmidt$)$.
\end{lem}
As an application of Lemma 1 recall that Proposition 1 in \cite{DP3} says that for $g \in C_0^{\infty}(\R)$ the operators $V g(H)$ and $V g(H_0)$ are trace class. Obviously, the same is true for $V(x + \tau, y)g(H)$ and we will use this fact below.
Consider the shift operator 
$$\uh: f(x, y) \longrightarrow f(x + \tau, y).$$
 Let $H_0 = (D_x - y)^2 + D_y + x, H = H_0 + V(x, y).$ It is clear that
$$[\uh, H_0] u = \uh H_0 u - H_0 \uh u = \uh (xu) - x\uh u = \tau\uh u,$$
hence $[\uh, H_0] = \tau \uh.$ Next
$$[\uh, V] = \uh(Vu) - V \uh u = V(x + \tau)\uh u - V \uh u = \Bigl(V(x+\tau, y) - V(x, y)\Bigr)\uh u.$$
Thus given a function $f \in C_0^{\infty}(\R)$, we get
$$[\uh, H] f(H) - [\uh, H_0] f(H_0) = \Bigl[ \tau + (V(x +\tau, y) - V(x, y))\Bigr] \uh f(H) - \tau \uh f(H_0)$$
$$= \tau \uh \Bigl(f(H) - f(H_0)\Bigr) + \Bigl(V(x +\tau, y) - V(x, y)\Bigr) \uh f(H).$$
\begin{prop} We have the equality
\begin{equation} \label{eq:2.1}
{\rm tr} \Bigl( [\uh, H] f(H) - [\uh, H_0] f(H_0)\Bigr) = 0.
\end{equation}

\end{prop}
\begin{proof} We write
$${\rm tr}\: \Bigl[ \uh H f(H) - \uh H_0 f(H_0) + H_0 \uh f(H_0) - H \uh f(H)\Bigr]$$
$$ = {\rm tr}\: \uh \Bigl(H f(H) - H_0 f(H_0)\Bigr) + {\rm tr} \Bigl( H_0 \uh f(H_0) - H \uh f(H)\Bigr) = (I) + (II).$$
For the term $(I)$, by using the cyclicity of the trace, we have
\begin{equation} \label{eq:2.2}
(I) = {\rm tr}\: \Bigl((H f(H) - H_0 f(H_0)) \uh\Bigr) = {\rm tr}\: \Bigl(f(H) H - f(H_0) H_0\Bigr) \uh.
\end{equation}
On the other hand,
$$(II) = {\rm tr}\: \Bigl((H_0- H) \uh f(H_0) \Bigr) + {\rm tr}\: \Bigl[ H \uh \Bigl(f(H_0) - f(H)\Bigr) \Bigr] = (II_1) + (II_2).$$
and we justify below the trace class properties of the operators $(II_1)$ and $(II_2)$.
For $(II_1)$ we write
$$- (II_1) = V \uh f(H_0) = \uh [ \uh^{-1} V \uh ] f(H_0) = \uh V(x - \tau, y) f(H_0)$$
and the operator on the right hand side is trace class. 

It is easy to see that the operator $\Bigl(f(H_0) - f(H)\Bigr)(H + \ii)$ is trace class since
$$\Bigl(f(H_0) - f(H)\Bigr)(H + \ii)= \Bigl[f(H_0)(H_0 + \ii) - f(H)(H + \ii)\Bigr] + f(H_0) V,$$
where on the right hand side we have a sum of two trace class operators. The same argument shows that the operator $H (f(H_0) - f(H))$ is trace class. Next we show that the operator $H (f(H_0) - f(H)) (H + \ii)$ is trace class.
To do this, we write
$$H (f(H_0) - f(H)) (H + \ii) = \Bigl(H_0 f(H_0) (H_0 + \ii) - H f(H)(H + \ii)\Bigr) + V f(H_0) (H_0 + \ii)$$
$$+ V f(H_0) V + H_0 f(H_0) V$$
and the four operators on the right hand side are trace class. This implies that
$H\uh (f(H_0) - f(H))(H + \ii)$ is trace class since the commutator $[H , \uh]$ is a bounded operator.
After these preparations we write
$$(II_2) = H \uh (f(H_0) - f(H)) = \uh H (f(H_0) - f(H)) + [H, \uh] (f(H_0) - f(H))$$
which obviously is trace class. Exploiting the trace class properties, we can write
$$(II_2) = {\rm tr}\: \Bigl[ H \uh (f(H_0) - f(H)(H + \ii)(H + \ii)^{-1}\Bigr]$$
$$ = {\rm tr}\: \Bigl[ (H + \ii)^{-1}H \uh (f(H_0) - f(H))(H + \ii)\Bigr]$$
$$ = {\rm tr}\: \Bigl((f(H_0) - f(H))(H+ \ii)(H + \ii)^{-1} H \uh\Bigr) = {\rm tr}\: \Bigl((f(H_0) - f(H)) H \uh\Bigr).$$
Combining the above expressions, we get
$$(I) + (II_1) + (II_2) = {\rm tr}\: \Bigl((H_0 - H) \uh f(H_0)\Bigr) + {\rm tr}\: \Bigl(f(H_0)(H - H_0) \uh\Bigr)$$
$$= {\rm tr}\: \Bigl(-V \uh f(H_0)\Bigr) + {\rm tr}\: \Bigl(\uh f(H_0) V\Bigr).$$
It remains to show that ${\rm tr}\: \Bigl(V \uh f(H_0)\Bigr) = {\rm tr}\:\Bigl(\uh f(H_0) V \Bigr).$ To do this, choose a function $\chi \in C_0^{\infty}(\R^2)$ such that $\chi = 1$ for $|(x,y)|\leq 1$. For
$R > 0$ set
$$\chi_R(x, y) = \chi\Bigl(\frac{x}{R}, \frac{y}{R}\Bigr)$$
and consider 
$${\rm tr}\:\Bigl(V \uh f(H_0) \chi_R\Bigr) = {\rm tr}\: \Bigl(\uh f(H_0) V \chi_R\Bigr).$$
The operator $\chi_R$ converges strongly to identity as $ R \to \infty$ and applying the well known property of trace class operators (see for instance, Proposition 1 in \cite{DP3}), we conclude that
$${\rm tr}\:\Bigl(V \uh f(H_0)\Bigr) = {\rm tr}\: \Bigl(\uh f(H_0) V \Bigr)$$
and the proof is complete.
\end{proof}
{\bf Proof of Theorem 1.} According to Proposition 1, we have
\begin{equation} \label{eq:2.3}
{\rm tr}\:\Bigl(\uh (f(H) - f(H_0)\Bigr) = - {\rm tr}\: \Bigl(\frac{V(x + \tau, y) - V(x, y)}{\tau} \uh f(H)\Bigr).
\end{equation}
We take the limit $ \tau \to 0$ and observe that 
$$\uh \longrightarrow I,\: \frac{V(x + \tau, y) - V(x, y)}{\tau} \uh  \longrightarrow \partial_x V$$
 strongly. Since $(f(H) - f(H_0))$ is a trace class operator, applying once more the property of trace class operators, we get
$$\lim_{\tau \to 0} {\rm tr}\: \Bigl(\uh (f(H) - f(H_0)\Bigr) = {\rm tr}\: (f(H) - f(H_0)).$$
To treat the limit $\tau \to 0$ in the right hand term of (\ref{eq:2.3}), consider the function,
$$g_{\delta}(x,y) = \la x \ra ^{-2-\delta} \la y \ra ^{-1-\delta}$$
 $\delta > 0$ being the constant of (1.1). Following Lemma 1, the operator $g_{\delta}(H_0 + \ii)^{-2}$ is trace class. Hence
$$g_{\delta} f(H) = g_{\delta} (f(H) - f(H_0)) + g_{\delta}(H_0 + \ii)^{-2} (H_0 + \ii)^2 f(H_0)$$
is also a trace class operator.

 To treat the limit $\tau \to 0,$ we use the representation
$$\Bigl(\frac{V(x + \tau, y) - V(x, y)}{\tau} g_{\delta}^{-1}\Bigr) \Bigl[ g_{\delta}\uh g_{\delta}^{-1}] g_{\delta} f(H).$$
The operators in the brackets $\Bigl(...\Bigr)$, $[...]$ converge strongly as $\tau \to 0$ to $(\partial_x V) g_{\delta}^{-1}$ and $I$, respectively. Letting $\tau \to 0$, we obtain 
$$\lim_{\tau \to 0} {\rm tr}\: \Bigl(\frac{V(x + \tau, y) - V(x, y)}{\tau}\Bigr) \uh f(H) = {\rm tr}\: \Bigl((\partial_x V) f(H)\Bigr)$$
 and the proof is complete.
\begin{rem} The proof of Theorem $1$ works for operators $M =(D_x - C(y))^2 + D_y^2 + \epsilon x + V(x, y)$ with non-linear $C(y)$ assuming that we have an analog of Lemma $1$ for $H$ and $H_0$ replaced by $M$ and $M_0 = (D_x - C(y))^2 + D_y^2 + \epsilon x,$ respectively. Also we may examine the operators in $\R^3$ having the form
$$\Bigl(D_x + \frac{B}{2}y\Bigr)^2 + \Bigl(D_y - \frac{B}{2}x\Bigr)^2 + D_z^2 + \epsilon z + V(x, y, z)$$
applying the shift operator $U_{\tau}: \:f(x, y, z) \longrightarrow f(x, y , z + \tau).$ Some operators with magnetic potentials and Stark potential in $\R^n,\: n \geq 3,$ can be investigated by the same approach.

\end{rem}

Now consider the operators $H_0(h) = (hD_x - By)^2 + h^2D_y^2 + \epsilon x,\: H(h) = H_0(h) + V(x, y),\: h > 0.$
Under the assumption (1.1) for $V(x, y)$ we have the statement of Lemma 1 for $H_0$ replaced by $H_0(h)$. Moreover, the operators $Vg(H(h))$ and $Vg(H_0(h))$ are trace class for every $g \in C_0^{\infty}(\R).$ Thus for every $f \in C_0^{\infty}(\R)$ the operator $f(H(h)) - f(H_0(h))$ is trace class and we can define the spectral shift function $\xi_h = \xi_h(\lambda, B , \epsilon)$ modulo constant by the formula
$$\langle \xi_h', f \rangle = {\rm tr}\: \Bigl(f(H(h) - f(H_0)\Bigr),\: f \in C_0^{\infty}(\R).$$
Under the assumption of Theorem 1, we obtain repeating the proof of (1.2) the representation
\begin{equation} \label{eq:2.4}
{\rm tr}\Bigl(f(H(h)) - f(H_0(h))\Bigr) = -\frac{1}{\epsilon} {\rm tr} \Bigr((\partial_x V) f(H(h))\Bigr).
\end{equation}

\section{ Embedded eigenvalues of $H$}


In this section we use the notation
$$L = H(0) = (D_x - By)^2 + D_y^2 + \epsilon x.$$

 Our purpose is to prove the following

\begin{thm}
There exists $C > 0$ such that $H$ has no eigenvalues $\lambda,\: |\lambda| \geq C.$  
\end{thm}

\begin{proof}
First notice that for every function $f \in C_0^{\infty}(\R)$ we have 

\begin{equation}\label{eq:3.1}
f(H)[\partial_x,H]f(H)= \epsilon f^2(H) + f(H)\partial_{x}V f(H).
\end{equation}
We will show the absence of embedded eigenvalues $\lambda > C > 0$. The case $\lambda < - C$ can be treated by the same argument. Assume that there exists a sequence of eigenvalues $\lambda_n \longrightarrow + \infty,\: \lambda_{n+1} > \lambda_n + 1,\: \forall n$ and let $H\varphi_n = \lambda_n \varphi_n,\: n \in \N$ with $(\varphi_i, \varphi_j) = \delta_{i,j}.$ Choose cut-off functions $f_n(t) \in C_0^{\infty}(\R)$ so that $f_n(\lambda_n) = 1,\: 0 \leq f_n(t) \leq 1$ and $f_n(t) = 0$ for $|t - \lambda_n| \geq 1/2$. It is clear that $f_n(H) \varphi_n = \varphi_n$ and
$$(\varphi_n, f_n(H)[\partial_x, H]f_n(H)\varphi_n) = 0,\: \forall n \in\N.$$
We wish to prove that for $n$ large enough we have
\begin{equation} \label{eq:3.2}
\Bigl|(\varphi_n, f_n(H) \pv f_n(H) \varphi_n)\Bigr| = \Bigl|(\varphi_n, \pv f_n(H) \varphi_n)\Bigr| \leq  \epsilon/2 
\end{equation}
which leads to a contradiction with (\ref{eq:3.1}) since $(\varphi_n, f^2_n(H) \varphi_n) = 1.$ Consider the operator
$$ f_n(H) = -\frac{1}{\pi}\int_{W_n}\bar{\partial} \tilde{f}_n(z)(z - H)^{-1} L(dz),$$
where $\tilde{f}_n(z)$ is an almost analytic continuation of $f_n$ with supp $\tilde{f}_n(z) \subset W_n$, $W_n = \{z \in \C:\: |z - \lambda_n| \leq 2/3\}$ is a complex neighborhood of $\lambda_n$ and 
$$\bar{\partial}\tilde{f}_n(z) = {\mathcal O}(|\Im z|^{\infty})$$
 uniformly with respect to $n$. Here $L(dz)$ is the Lebesgue measure in $\C.$ We write
$$(\varphi_n, \pv f_n(H)\varphi_n) = -\frac{1}{\pi}  \int_{W_n\cap \{|\Im z| \leq \eta\}}  \bar{\partial}\tilde{f}_n(z)(\varphi_n, (\pv) (z - H)^{-1}\varphi_n) L(dz)$$
$$  -\frac{1}{\pi}  \int_{W_n\cap \{|\Im z| > \eta\}}  \bar{\partial}\tilde{f}_n(z)(\varphi_n, (\pv - V_0)(z - H)^{-1}\varphi_n) L(dz) $$
$$ -\frac{1}{\pi}  \int_{W_n\cap \{|\Im z| > \eta\}}  \bar{\partial}\tilde{f}_n(z)(\varphi_n,  V_0(z - H)^{-1}\varphi_n) L(dz) = R_n + Q_n + S_n,$$
where $V_0(x, y) \in C_0^{\infty}(\R^2).$
We choose $\eta > 0$ small enough to arrange $|R_n| \leq \epsilon/6$ for all $n \in \N.$ Next we fix $0 < \eta < 1$ and we will estimate $Q_n$  and $S_n.$
For the resolvent $(z - L)^{-1}$ we will exploit the following 
\begin{prop} $($\cite{FK2}$)$  Let $f,\: g$ be bounded functions with compact support in $\R^2$. Then for every compact ${\mathcal K} \subset \R \setminus \{0\}$ we have 
$$\lim_{\lambda \to \pm \infty}  \|f (\lambda + \ii \gamma - L)^{-1} g\| = 0$$
uniformly for $\gamma \in {\mathcal K}.$
\end{prop}
We choose $V_0$ so that $\|\pv - V_0\|$ is sufficiently small in order to arrange $|Q_n| \leq \epsilon/6,\: \forall n \in \N.$
Now we pass to the estimation of $S_n.$ We have
\begin{equation} \label{eq:3.3}
V_0(z - H)^{-1} = V_0(z - L)^{-1} + V_0(z - L)^{-1}(V - V_1) (z - H)^{-1} + V_0(z - L)^{-1} V_1 (z - H)^{-1}.
\end{equation}
 We replace  $V_0(z - H)^{-1}$ in $S_n$ by the right hand side (\ref{eq:3.3}) choosing $V_1 \in C_0^{\infty}(\R^2).$
For the term involving $(V - V_1)$ in (\ref{eq:3.3}) we take $V_1$ so that $\|V - V_1\|$ is small enough, to obtain a term bounded by $\epsilon/18.$ Next we fix the potentials $V_0,\: V_1$ with compact support. By Proposition 2 setting $z = \lambda + \ii\gamma,\: \eta \leq |\gamma| \leq 1,$ we get
$$\|\bar{\partial}\tilde{f}_n(z)V_0( \lambda + \ii \gamma - L)^{-1} V_1 (H - z)^{-1}\|\leq C_2 \eta^{-1}\|V_0(\lambda + \ii \gamma - L)^{-1} V_1\| \leq \frac{9}{4 \pi^2}\frac{\epsilon}{18}$$
for $\Re z = \lambda \geq C_{\epsilon, \eta}.$ We choose $n \geq n_0 = n_0(\epsilon, \eta)$, so that $\Re z \geq C_{\epsilon, \eta}$ for $z \in W_n$ and $n \geq n_0$. Thus we can estimate the term involving $V_0(z - L)^{-1} V_1$ in (\ref{eq:3.3}) by $\epsilon/18.$ It remains to deal with the term containing $V_0 (z - L)^{-1}.$ Let $\psi(x, y) \in C_0^{\infty}(\R^2)$ be a cut-off function such that $\psi = 1$ on the support of $V_0$. We write
$$\psi V_0 (z - L)^{-1} = V_0 (z - L)^{-1} \psi - V_0 (z - L)^{-1} [(D_x - By)^2 + D_y^2, \psi] (z -L)^{-1}$$
$$ = V_0 (z - L)^{-1} \psi - V_0 (z - L)^{-1} \psi_1 [(D_x - By)^2 + D_y^2, \psi] (z -L)^{-1},$$
where $\psi_1$ is a cut-off function equal to 1 on the support of $\psi.$ For $n$ large enough we will have $\Re z = \lambda \geq C'_{\epsilon, \eta}$ for $z \in {\rm supp}\: W_n$ and can treat $V_0 (z - L)^{-1} \psi$ and $V_0 (z - L)^{-1}\psi_1$ as above. On the other hand,
\begin{equation} \label{eq:3.4}
[(D_x - By)^2 + D_y^2, \psi] = -2 \ii \partial_x \psi (D_x - By) - 2 \ii\partial_y \psi D_y - \Delta_{x, y} \psi
\end{equation}
and the operators $\partial_x \psi(D_x - By) (z - L)^{-1}$ and $\partial_y \psi D_y (z - L)^{-1}$ are bounded by $C\eta^{-1}$ for $|\Im z| \geq \eta.$  Indeed, we have
$$(z - L) = (\ii - L)^{-1}[ I +  (\ii - z) (z - L)^{-1}]$$
and it suffices to show that $\partial_x \psi(D_x - By) (\ii - L)^{-1}$ and $\partial_y \psi D_y (\ii - L)^{-1}$ are bounded.
Next, $(\ii - L)^{-1}$ is a pseudodifferential operator 
and the symbol of the pseudodifferential operator $(D_x - By)(\ii - L)^{-1}$ becomes
$$\frac{\xi - By}{\ii - (\xi - By)^2 - \eta^2 - \epsilon x} -\frac{\ii B\eta}{(\ii - (\xi - By)^2 - \eta^2 - \epsilon x)^2}.$$
From the well known results for the $L^2$ boundedness of pseudodifferential operators (see \cite{DS}) we deduce that (\ref{eq:3.4}) is bounded by $C|\Im z|^{-1}.$ Consequently, applying Proposition 2 once more, we can arrange the norm of the operator
$$V_0 (z - L)^{-1} \psi_1 [(D_x - By)^2 + D_y^2, \psi] (z -L)^{-1}$$
to be sufficiently small for $z \in W_n,\: |\Im z | \geq \eta$ and $n \geq n_1 > n_0$.
Combining this with the previous estimates, we get
$|S_n| \leq \epsilon/6$, hence $|R_n + Q_n + S_n| \leq \epsilon/2$ for $n$ large enough.
This implies (\ref{eq:3.2}) and the proof is complete.
\end{proof}
\begin{cor} Assume in addition to $(1.1)$ that $\partial_x^2 V \in C_0(\R^2) \cap L^{\infty}(\R^2)$. Then $H$ has at most finite number of embedded eigenvalues in $\R.$
\end{cor}
This result follows from Theorem 2 and Corollary $1$ in \cite{DP3} which guarantees  that $H$ has at most finite number of embedded eigenvalues in every interval $[a, b] \subset \R.$  The conjecture is that $H$ has no embedded eigenvalues on $\R$.

\section{Asymptotics of the spectral shift function}

 Our purpose in this section is to apply Theorem 1 and (2.4) to give a Weyl type  asymptotics with optimal remainder estimates for the spectral shift function $\xi_h(\lambda):=\xi(\lambda;H(h), H_0(h))$  corresponding to the operators
$$H(h)=(h D_x- y)^2+h^2D_y^2+ x,\:H_0(h) = H(h) + V(x, y),\: h > 0.$$
For simplicity of the exposition in this section we assume that $B = \epsilon = 1.$ 
Let $p_2(x, y, \zeta, \eta) = (\zeta - y)^2 + \eta^2 + x + V(x, y).$
For the analysis of $\xi_h(\lambda)$ we need the following theorems.

\begin{thm}
Let $\psi\in C^ \infty_0(\R^2)$ and let $f\in C^\infty_0(]0,+\infty[;\R)$. Then we have
\begin{equation}\label{eq:4.1}
{\rm tr}\Bigl[ \psi f(H(h))\Bigr]\sim\sum_{j=0}^ \infty a_j h^{j-2},\,\, h\searrow 0,
\end{equation}
with
\begin{equation}\label{eq:4.2}
a_0={1\over (2\pi)^ 2} \iint\psi(x,y) f(p_2(x,y,\zeta, \eta)) dx dyd\zeta d\eta.
\end{equation}
\end{thm}

\medskip

\begin {thm}  Assume that $\psi\in C^\infty_0(\R^2)$.  Let   $f \in C^\infty_0(\lbrack E_0, E_1 \lbrack)$ and $\theta\in C^\infty_0(\rbrack-{1\over C_0}, {1\over C_0}\lbrack;\:{\R}),\:\: \theta = 1$ in a neighborhood of $0$. Assume that if $p_2(x,y, \zeta, \eta) = \tau,\: \tau \in [E_0, E_1]$, then $dp_2 \neq 0.$ Then  there exists $C_0>0$ such that for all $N, m\in \N$ there exists $h_0>0$ such that
\begin{equation} \label{eq:4.3}
{\tr}\Bigl(\psi\breve\theta_h(\tau-H(h))f(H(h))\Bigr)
=(2 \pi h)^{-2} \Bigl( f(\tau)\sum_{j=0}^{N -1}\gamma_j(\tau) h^j +{\mathcal O}(h^N\langle \tau\rangle^{-m})\Bigr),
\end{equation}
uniformly with respect to $\tau \in \R$ and  $h\in ]0,h_0]$,
where
$$\gamma_0(\tau) = -(2\pi \ii)^{-1}\int\int_{\R^{4}}
\psi(x,y)\Bigl((\tau + \ii 0 - p_2(x,y,\zeta, \eta))^{-1}- (\tau - \ii 0 - p_2(x,y, \zeta, \eta))^{-1}\Bigr) dxdy d\zeta d\eta.$$
Here $$\breve\theta_h(\tau)= (2\pi h)^{-1} \int e^{i\tau t/h} \theta(t)dt.$$

\end{thm}
\medskip

{\it Proof of Theorem $3$ and Theorem $4$.} Here and below $\psi \prec \varphi$ means that $\varphi(x)=1$ on the support of $\psi$.
Let $G$  $\in C^{\infty}_0(\R^2)$ with  $\psi  \prec  G$.
 Introduce the operator
$$ \tilde H(h) =(hD_x-G(x,y)y)^2+h^2D_y^2+G(x,y)x + V(x,y),$$
and set
$$I ={\tr}\Bigl[ \psi \Bigl(f(H(h))-f(\tilde H(h)\Bigr)\Bigr].$$
 Let $\tilde{f}(z) \in C_0^{\infty}(\C)$ be an almost analytic continuation of $f$ with $\bar{\partial}_z \tilde{f}(z) ={\mathcal O}(\vert \Im z\vert^\infty)$ .
 From Helffer-Sj\"ostrand formula it follows    that
$$
I =\frac{1}{\pi} \int  \bar{\partial}_z \tilde f(z) 
\tr\Bigl[\psi \Bigl((z-\tilde H(h))^{-1}- (z- H(h))^{-1}\Bigr) \Bigr] L(dz),$$
where $L(dz)$ denotes the Lebesgue measure on $\C.$

 Let $\psi_1\in C^\infty(\R^2)$ be a function with $ \psi_1=1$ near 
${\rm supp}\:\:(1-G)$ and $\psi_1=0$ near ${\rm supp}\:\psi$, and let $\tilde \psi\in C^\infty_0(\R^2)$ be equal to  one near ${\rm supp}(\nabla\psi_1)$ and $\tilde\psi=0$
near ${\rm supp}\:\psi$. A simple calculus shows that $\tilde H(h)-H(h)= \psi_1 (\tilde H(h)-H(h))$ and $[\tilde H(h), \psi_1]=\tilde\psi[\tilde H(h), \psi_1]\tilde H$.
Then
\begin{equation}\label{eq:4.4}
\psi\Bigl((z-\tilde H(h))^{-1}-(z-H(h))^{-1}\Bigr)= \psi
(z- \tilde H(h))^{-1} \psi_1(\tilde H(h)-H(h))(z- H(h))^{-1}
\end{equation}
$$=\psi(z-\tilde H(h))^{-1}\tilde\psi[\tilde H(h),\psi_1](z-\tilde H(h))^{-1}  (\tilde H(h)-H(h)) (z-H(h))^{-1}.$$
Let $\chi_1,...,\chi_N \in C_0^\infty(\R^2;\lbrack 0,1\rbrack)$ 
with $\psi_1\prec \chi_1\prec...\prec \chi_N$ and $\chi_i\tilde\psi=0,\: i = 1,...,N$. 
By using  the equalities $\chi_1 \psi_1 =...=\chi_N\psi_1=\psi_1,\: \chi_k\: \tilde \psi=0$,  $\chi_{k-1}[\chi_k,\tilde H(h)]=0$ and the fact that
 $$[\chi_k,(z-\tilde H(h))^{-1}]=
(z-\tilde H(h))^{-1} [\chi_k,\tilde H(h)](z-\tilde H(h))^{-1},$$
 we get
$$       \psi(z-\tilde H(h))^{-1}\tilde\psi[\tilde H(h),\psi_1]  $$
$$ =\psi(z-\tilde H(h))^{-1}
[\chi_1,\tilde H(h)](z-\tilde H(h))^{-1}...[\chi_N,\tilde H(h)]
(z-\tilde H(h))^{-1} \tilde\psi[\tilde H(h),\psi_1]     =:L_N(h).$$

Here
$$L_N(h) ={\mathcal O}_N(1)\Bigl( {h^N\over \vert\Im z\vert^N}\Bigr): H^{s}(\R^2)\rightarrow H^{s+N}(\R^2),$$
where we equip $H^N(\R^2)$ with the $h$-dependent norm $\Vert \langle hD\rangle^N u\Vert_{L^2}$.  Choose  $N > 2$ and let $s=-N$.
According to  Theorem 9.4 of \cite{DS}, we have
$$\Bigl\|\Bigr(-h^2\Delta+ 1\Bigr)^{-N/2}\tilde\psi\Bigr\|_{\tr}={\mathcal O}(h^{-2}).$$
Then
\begin{equation} \label{eq:4.5} 
\|    \psi(z-\tilde H(h))^{-1}\tilde\psi[\tilde H(h),\psi_1]  \tilde{\psi}\|_{\tr} =
 \Bigl\| L_{N}(h)     \Bigr(-h^2\Delta + 1\Bigr)^{N/2}  \Bigr(-h^2\Delta + 1\Bigr)^{-N/2}     
  \tilde\psi \Bigr\|_{\tr}
\end{equation}
$$\leq C \Bigl\|\Bigr(-h^2\Delta+ 1\Bigr)^{-N/2}  \tilde\psi\Bigr\|_{\tr} \Bigr({h^N\over \vert\Im z\vert^N}\Bigr) \leq C_1 \Bigl({h^{N-2}\over \vert\Im z\vert^N}\Bigr).$$
Combining this with (4.4) and using the fact that 
$$\Vert (z-\tilde H(h))^{-1}  (\tilde H(h)-H(h)) (z-H(h))^{-1}\Vert=\Vert (z-\tilde H(h))^{-1}-(z-H(h))^{-1}\Vert=
{\mathcal O} \Bigl(\vert \Im z\vert^{-1}\Bigr),$$
we obtain
$$\Bigl\|\psi\Bigl((z-\tilde H(h))^{-1}-(z-H(h))^{-1}\Bigr)\Bigr\|_{{\rm tr}}={\mathcal O} \Bigl({h^{N-2}\over \vert\Im z\vert^{N+1}}\Bigr).$$
Since  $\bar{\partial}_z \tilde{f}(z) = {\mathcal O}(|\Im z|^\infty)$, we have 
$$I= {\mathcal O}(h^\infty).$$

Summing up, we have proved that
\begin{equation} \label{eq:4.6}
{\tr}\Bigl(\psi f(H(h))\Bigr)
={\tr}\Bigl(\psi f(\tilde H(h))\Bigr) +{\mathcal O}(h^\infty).
\end{equation}
In the same way, we obtain
\begin{equation}\label{eq:4.7}
{\tr}\Bigl(\psi\breve\theta_h(\tau-H(h))f(H(h))\Bigr)
={\tr}\Bigl( \psi
 \breve\theta_h(\tau-\tilde H(h))f(\tilde H
(h))\Bigr) + {\mathcal O}(h^\infty).
\end{equation}

The operator $\tilde H(h)$ is elliptic semi-bounded $h$-pseudodifferential operator, so Theorem 3 and Theorem 4  follow from the $h$-pseudodifferential calculus and the analysis of elliptic operators in Chapters 8, 9, 12 in \cite{DS} (see also \cite{R}). The calculus of the leading terms is given by  trivial modification of the argument of Section 7 in  \cite{DP1} and we omit the details.
\hfill{$\Box$}\\
\begin{rem}
Notice that $dp_2\not=0$ on $p_2=\tau$ is equivalent to
\begin{equation} \label{eq:4.8}
 \nabla_{x,y}(x+V(x,y))\not=0,\,\,\,{\rm on}\,\,\,\, \{(x,y);\,\,\, x+V(x,y)=\tau\}.
\end{equation}
\end{rem}

Now we will apply  Theorem 3 and Theorem 4  to obtain a 
Weyl-type asymptotics for $\xi_h(\lambda)$ when $h \searrow 0.$ 
\begin{thm}
Assume that $V\in C^\infty_0(\R^2)$ and suppose that $(\ref{eq:4.8})$  holds for  $\tau =\lambda_1,\lambda_2$.
Then there exists $h_0 > 0$ such that for $h \in ]0,h_0]$ we have
\begin{equation}\label{eq:4.9}
\xi_h(\lambda_2)-\xi_h(\lambda_1)=(2\pi h)^{-2} (c_0(\lambda_2)-c_0(\lambda_1))+{\mathcal O}(h^{-1}),
\end{equation}
where  
\begin{equation}\label{eq:4.10}
c_0(\lambda)= -\pi \:\int_{\R^2} \partial_x V(x,y)(\lambda - x - V(x,y))_{+} dxdy.
\end{equation}
\end{thm}

\begin{proof}

Choose a large constant $M$ such that 
$$M\geq \Vert \partial_x V\Vert_\infty+1.$$
Let $\psi\in C^\infty_0(\R^2;[0,1])$ with $\partial_x V \prec  \psi^2$.
According to (\ref{eq:2.4}), by using the cyclicity of the trace, we get
$$\langle \xi_h', f \rangle = {\rm tr}\: \Bigl(f(H(h)) - f(H_0(h)\Bigr) = -{\rm tr}\Bigl((\partial_x V) f(H(h))\Bigr)$$
$$={\rm tr}\Bigl((M-\partial_x V )^{1/2} \psi f(H(h) \psi (M-\partial_x V )^{1/2} \Bigr)-M
{\rm tr}\Bigl( \psi f(H(h))\psi\Bigr)$$
$$=:\langle \xi'_1,f\rangle-\langle \xi'_2,f\rangle.$$
Since 
$$f \rightarrow {\rm tr}\Bigl((M-\partial_{x_1}V )^{1/2} \psi f(H(h) \psi (M-\partial_{x_1}V )^{1/2} \Bigr)$$
 and 
$$f \rightarrow M{\rm tr}\Bigl( \psi f(H(h))\psi\Bigr)$$
are positive functions for $f\geq 0$, we deduce that the functions $\lambda \rightarrow \xi_i(\lambda), \: i = 1,2$ are monotonic. 

Consequently, we may apply Tauberian arguments for the analysis of the asymptotics of $\xi_i(\lambda), \: i = 1,2.$ 
We treat below $\xi_2(\lambda).$ Let $\varphi \in C_0^{\infty}(\R),\: \varphi \geq 0,$  and suppose that  (\ref{eq:4.8}) holds for all 
$\tau\in {\rm supp}\: \varphi.$ Consider the function
$$F_{\varphi}(\lambda) = \int_{-\infty}^\lambda \xi_2'(\mu) \varphi(\mu) d\mu.$$
Applying (4.3) with $N=1$ and $m=2$, we obtain
\begin{equation}\label{eq:4.11}
\frac{d}{d\lambda} (\breve{\theta}_h \ast F_{\varphi})(\lambda)= \int \breve{\theta}_h(\lambda - \mu) \xi_2'(\mu) \varphi(\mu) d\mu =(2\pi h)^{-2} \Bigl(\varphi(\lambda)\gamma_0(\lambda)+{\mathcal O}\Bigl(\frac{h}{ \langle \lambda\rangle^2}\Bigr)\Bigr).
\end{equation}

We integrate from $-\infty $ to $\lambda$  and we get
\begin{equation}\label{eq:4.12}
\int\Bigl(\int_{-\infty}^\lambda \breve{\theta}_h(\lambda'-\mu) d\lambda' \Bigr)\xi'_2(\mu) \varphi(\mu) d\mu
\end{equation}
$$= {1\over (2\pi h)^2} \Bigl( \int\int_{ p_2\leq \lambda} M \psi^2(x,y) \varphi(p_2) dxdyd\eta d\zeta+{\mathcal O}(h)\Bigr).$$
In the following we choose $\theta\in C^\infty_0(\R)$ with $\breve{\theta}_h\geq 0$. Let $h\breve{\theta}_h(0) = \frac{1}{2\pi}\int_\R \theta(u) du \geq 2C_1 > 0.$
Therefore, it follows that there exist $C_2 > 0$ such that
 $$\vert t\vert <{h\over C_2}\Longrightarrow  h\breve\theta_h(t)\geq C_1.$$
Combining this with the fact that $\breve\theta_h\geq 0$, and using $\langle \xi'_2, f\rangle\geq 0$ for $f\geq 0$,  we obtain
\begin{equation}\label{eq:4.13}
C_1 \int_{\{\vert \lambda -\mu\vert<{h\over C_0}\}} \xi'_2(\mu)\varphi(\mu)d\mu \leq 
h\int_{\{\vert \lambda -\mu\vert<{h\over C_0}\}}\breve{\theta}_h(\lambda - \mu) \xi_2'(\mu) \varphi(\mu) d\mu$$
$$\leq h \int_{\R}\breve{\theta}_h(\lambda - \mu) \xi_2'(\mu) \varphi(\mu) d\mu=
h\frac{d}{d\lambda} (\breve{\theta}_h \ast F_{\varphi})(\lambda) ={\mathcal O}(h^{-1}),
\end{equation}
uniformly with respect to  $\lambda\in\R$. On the other hand, a simple calculus shows that
\begin{equation}\label{eq:4.14}
\int_{-\infty}^\lambda \breve{\theta}_h(\lambda'-\mu) d\lambda' =
\int_{-\infty}^{{\lambda-\mu\over h}}  \breve{\theta}_1(t) dt=
{\bf 1}_{]-\infty, \lambda[}(\mu) +
{\mathcal  O}\Bigl(\big\langle {\frac{\lambda-\mu}{h}}\big\rangle^{-\infty}\Bigr).
\end{equation}
Indeed, for $\mu < \lambda$ and all $k \in \N$ we have
$$\int_{-\infty}^{\frac{\lambda - \mu}{h}} \breve{\theta}_1(t) dt -1 = -\int_{\frac{\lambda - \mu}{h}}^{\infty} t^k\breve{\theta}_1(t) \frac{1}{t^k}dt$$
and 
$$\int_{\frac{\lambda - \mu}{h}}^{\infty} t^k\breve{\theta}_1(t) \frac{1}{t^k}dt
 \leq \Bigl(\frac{\lambda - \mu}{h}\Bigr)^{-k} \int_{\R} t^k \breve{\theta}_1(t) dt.$$
A similar argument works for $\mu > \lambda.$
From (\ref{eq:4.13}) we have for $k \geq 2$ the estimate
\begin{eqnarray}\label{eq:4.15}
\int_{\R}\big \langle \frac{\lambda-\mu}{h}\big\rangle^{-k}\xi_2'(\mu) \varphi(\mu)d\mu = \sum_{m = -\infty}^{\infty}\int_{\frac{m}{C_0} \leq \frac{\mu - \lambda}{h} < \frac{m+1}{C_0}} \big\langle \frac{\lambda-\mu}{h}\big\rangle^{-k}\xi_2'(\mu) \varphi(\mu)d\mu\\ \nonumber
\leq \sum_{m = 0}^{\infty} \Bigl(1 + \Bigl(\frac{m}{C_0}\Bigr)^2\Bigr)^{-k/2} \int_{\lambda + \frac{m h}{C_0}}^{\lambda + \frac{(m+1)h}{C_0}}\xi_2'(\mu) \varphi(\mu) d\mu\\ \nonumber
+ \sum_{m =-\infty}^{-1} \Bigl(1 + \Bigl(\frac{|m +1|}{C_0}\Bigr)^2\Bigr)^{-k/2} \int_{\lambda + \frac{m h}{C_0}}^{\lambda + \frac{(m+1)h}{C_0}}\xi_2'(\mu) \varphi(\mu) d\mu
\leq  \sum_{m = -\infty}^{\infty}  \frac{1}{(C_0 + |m|)^k} {\mathcal O}(h^{-1}),
\end{eqnarray}
where in the last inequality at the right hand side we used the fact that (\ref{eq:4.13}) holds uniformly with respect to $\lambda \in \R$ and we can estimate the integrals involving $\xi_2'(\mu) \varphi(\mu)$ by ${\mathcal O}(h^{-1})$ uniformly with respect to $m$.


Inserting the right hand side of (\ref{eq:4.14})   in the left hand side of (\ref{eq:4.12})  and using (\ref{eq:4.15}), we get
$$ F_{\varphi}(\lambda) =  (2\pi h)^{-2} \Bigl( \int\int_{ p_2\leq \lambda}  M\psi^2(x,y) \varphi(p_2) dxdyd\eta d\zeta+{\mathcal O}(h)\Bigr).$$

We apply the same argument for $\xi_1(h)$ and introduce the function
$$G_{\varphi}(\lambda) = \int_{-\infty}^\lambda \xi_1'(\mu)\varphi(\mu) d\mu.$$
 Replacing the function $\psi$ by $(M -\partial_x V)^{1/2} \psi$, we
get
$$G_{\varphi}(\lambda) =
{1\over (2\pi h)^2} \Bigl( \int\int_{ p_2\leq \lambda} ( M- \partial_x V)\psi^2(x,y) \varphi(p_2) dxdyd\eta d\zeta+{\mathcal O}(h)\Bigr).$$
Since $\xi_h = \xi_1 - \xi_2$, the above results yield
\begin{equation}\label{eq:4.16}
M_\varphi(\lambda)=\int_{-\infty}^\lambda \xi_h'(\mu) \varphi(\mu) d\mu=
{1\over (2\pi h)^2} \Bigl( \int\int_{ p_2\leq \lambda}  -\partial_xV(x,y) \varphi(p_2) dxdyd\eta d\zeta+{\mathcal O}(h)\Bigr).
\end{equation}
Now, we are ready to prove Theorem 5. Assume that $\lambda_1<\lambda_2$, and let $\epsilon>0$ be small enough. Let
$\varphi_1, \varphi_2, \varphi_3\in C^\infty_0(]\lambda_1-\epsilon,\lambda_2+\epsilon[)$ with $\varphi_1+\varphi_2+\varphi_3=1$ on $[\lambda_1,\lambda_2]$,
 ${\rm supp}\:\varphi_1\subset
]\lambda_1-\epsilon,\lambda_1+\epsilon[$, ${\rm supp}\:\varphi_2\subset
]\lambda_2-\epsilon,\lambda_2+\epsilon[$ and  ${\rm supp}\:\varphi_3\subset
]\lambda_1,\lambda_2[$. We choose $\epsilon$ small enough so that (\ref{eq:4.8})
  holds for all $\tau\in ]\lambda_1-\epsilon,\lambda_1+\epsilon[\:\cup \:]\lambda_2-\epsilon,\lambda_2+\epsilon[$.
We write
$$\xi_h(\lambda_2)-\xi_h(\lambda_1)=\int_{\lambda_1}^{\lambda_2} (\varphi_1+\varphi_2+\varphi_3)(\lambda) \xi_h'(\lambda)d\lambda$$
$$= M_{\varphi_2}(\lambda_2) + M_{\varphi_1}(\lambda_2) -M_{\varphi_2}(\lambda_1) -M_{\varphi_1}(\lambda_1) -{\rm tr}(\partial_x V\varphi_3(H)),$$
where for the function $\varphi_3$ we have exploited (\ref{eq:2.4}). Next for the term involving $\varphi_3$ we
apply Theorem 3 and obtain
$${\rm tr}\: (\pv \varphi_3(H)) = \frac{1}{(2 \pi h)^{2}} \int\int \partial_x V \varphi_3(p_2) dx dy d\zeta d\eta + {\mathcal O}(h^{-1}).$$
For $M_{\varphi_{1}}(\lambda_i)$ and $M_{\varphi_{2}}(\lambda_i),\: i = 1,2,$ we exploit the above argument and we deduce the asymptotics taking into account (\ref{eq:4.16}). Summing the terms involving $\varphi_j,\:j = 1,2,3$, we conclude that
$$\xi_h(\lambda_2) - \xi_h(\lambda_1) = (2 \pi h)^{-2} d(\lambda_2, \lambda_1) + {\mathcal O}(h^{-1}).$$
For the leading term we have
$$d(\lambda_2, \lambda_1) = \int\int_{\lambda_1 \leq p_2 \leq \lambda_2} -\pv (x,y) \Bigl(\varphi_1(p_2) + \varphi_2(p_2)+ \varphi_3(p_2)\Bigr) dx dy d\zeta d\eta $$
 $$= -\int\int_{ p_2 \leq \lambda_2} \pv (x, y)  dx dy d\zeta d\eta + \int\int_{p_2 \leq \lambda_1} \pv (x, y) dx dy d\zeta d\eta.$$
Finally, notice that
$$c_0(\lambda) = -\int\int_{p_2 \leq \lambda} \pv(x, y) dx dy d\zeta d\eta = - \int_{\R^2} \pv(x, y)\Bigl(\int_{(\zeta - y)^2 + \eta^2 \leq (\lambda - x - V(x,y))_+} d\zeta d\eta\Bigr)dxdy$$
$$ =-\pi \int_{\R^2} \pv(x, y) (\lambda - x - V(x, y))_+ dx dy$$
and the proof of Theorem 5 is complete.
\end{proof}
\begin{rem}
If $\lambda \gg 1$ is large enough (resp. $\lambda \ll -1$) then on ${\rm supp}\:(\partial_x V)$, we have
$$(\lambda-x-V)_+=\lambda-x-V,\,\,\, ({\rm resp.}\,\, (\lambda-x-V)_+=0).$$
Consequently, 
$$c_0(\lambda)=- \pi\int_{\R^2} V(x,y)dxdy,\,\,\,\,{\rm for}\,\,\,  \lambda \gg 1,$$
and
$$c_0(\lambda)=0,\,\,\, {\rm for}\,\,\, \lambda \ll -1.$$
In particular, if we normalize $\xi_h(\lambda)$ by $\lim_{\lambda \rightarrow -\infty}\xi_h(\lambda)=0$, we get
$$\xi_h(\lambda)=(2\pi h)^{-2} c_0(\lambda)+{\mathcal O}(h^{-1}).$$
\end{rem}

\begin{rem} The results of this section can be generalized for potentials $V(x, y)$ for which there exists $\delta_1 \in \R$ such that supp $V \subset \{(x, y) \in \R^2:\: x \geq \delta_1\}$ by using the techniques in \cite{DP1}. For simplicity we treated the case of $V \in C_0^{\infty}(\R^2)$ to avoid the complications caused by the calculus of pseudodifferential operators.

\end{rem}

{\footnotesize

\end{document}